\documentclass[10pt,conference]{IEEEtran}

\IEEEoverridecommandlockouts

\usepackage{textcomp}
\usepackage{xcolor}
\usepackage{graphicx} 
\usepackage{cite}
\usepackage{graphicx}
\usepackage{amsmath}
\usepackage{amssymb}
\usepackage{amsfonts}
\usepackage{array}

\usepackage{color}
\usepackage{algorithm, algorithmic}
\usepackage{booktabs} 
\usepackage{multirow}
\usepackage{url}
\usepackage{amsthm} 

\newtheorem{proposition}{Proposition}
\newtheorem{theorem}{Theorem}

\usepackage{subcaption}
\usepackage{graphicx,bm,dsfont}

\newcommand{\todom}[1]{} 

\def\BibTeX{{\rm B\kern-.05em{\sc i\kern-.025em b}\kern-.08em
    T\kern-.1667em\lower.7ex\hbox{E}\kern-.125emX}}

\begin{document}

\title{Distributed Activity Detection for Cell-Free Hybrid Near-Far Field Communications\\
}

\author{
    \IEEEauthorblockN{Jingreng Lei\IEEEauthorrefmark{1}\IEEEauthorrefmark{2}, Yang Li\IEEEauthorrefmark{2}\thanks{The work of Yang Li was supported in part by Guangdong Basic and Applied Basic Research Foundation under Grant 2025A1515011658, in part by the National Natural Science Foundation of China (NSFC) under Grant 62101349, and in part by Guangdong Research Team for Communication and Sensing Integrated with Intelligent Computing (Project No. 2024KCXTD047). The work of  Ya-Feng Liu was supported  by the NSFC under Grant 12371314.}, Zeyi Ren\IEEEauthorrefmark{1}, Qingfeng Lin\IEEEauthorrefmark{1}, Ziyue Wang\IEEEauthorrefmark{3}, Ya-Feng Liu\IEEEauthorrefmark{4}, and Yik-Chung Wu\IEEEauthorrefmark{1}}
    
    \IEEEauthorblockA{\IEEEauthorrefmark{1}Department of Electrical and Electronic Engineering, The University of Hong Kong, Hong Kong}

    \IEEEauthorblockA{\IEEEauthorrefmark{2}School of Computing and Information Technology,
Great Bay University, Dongguan, China}
    \IEEEauthorblockA{\IEEEauthorrefmark{3}LSEC, ICMSEC, AMSS, Chinese Academy of Sciences, Beijing, China}
    \IEEEauthorblockA{\IEEEauthorrefmark{4} School of Mathematical Sciences, Beijing University of Posts and Telecommunications, Beijing, China}
      \IEEEauthorblockA{Emails: \{leijr, renzeyi, qflin, ycwu\}@eee.hku.hk, liyang@gbu.edu.cn, ziyuewang@lsec.cc.ac.cn, yafengliu@bupt.edu.cn}
    
}

\maketitle

\begin{abstract} 
A great amount of endeavor has recently been devoted to activity detection for massive machine-type communications in cell-free massive MIMO. However, in practice, as the number of antennas at the  access points (APs) increases, the Rayleigh distance that separates the near-field and far-field regions also expands, rendering the conventional assumption of far-field propagation alone impractical. To address this challenge, this paper considers a hybrid near-far field activity detection in cell-free massive MIMO, and establishes a  covariance-based formulation, which facilitates the development of a distributed algorithm to alleviate the computational burden at the  central processing unit (CPU). Specifically,  each AP performs local activity detection for the
devices and then transmits the detection result to the CPU for
further processing. In particular, a novel coordinate descent algorithm based on the Sherman-Morrison-Woodbury update with Taylor expansion is proposed to handle the local
detection problem at each AP. Moreover, we theoretically analyze how the hybrid near-far field channels affect the detection performance. Simulation results validate the theoretical analysis and demonstrate the superior performance of the proposed  approach compared with existing approaches.
\end{abstract}

\begin{IEEEkeywords}
Cell-free massive multiple-input multiple-output, distributed activity detection, grant-free random access, hybrid near-far field communications, machine-type communications.
\end{IEEEkeywords}

\section{Introduction}
With the rapid development of the Internet of Things (IoT), massive machine-type communications (mMTC) are expected to play a crucial role for the sixth-generation (6G) vision of ubiquitous connectivity\cite{Yafeng}. To meet the stringent low-latency requirements, grant-free random access has been demonstrated as a promising solution compared to grant-based schemes\cite{shahab2020grant}.

Activity detection is a crucial task during the grant-free random access phase. Mathematically, current studies on activity detection can be broadly divided into two lines. In the first line of research, by exploiting the sporadic nature of mMTC, the compressed sensing (CS)-based method is proposed to solve the joint activity detection and channel estimation problem~\cite{yangge3,gao2023compressive}. Another line of research, known as the covariance-based approach, identifies active devices by leveraging the statistical properties of channels based on the covariance matrix without estimating the channels\cite{li2022asynchronous,lin2024intelligent}. Both theoretically and empirically, it has been demonstrated that the covariance-based approach generally outperforms the \mbox{CS-based} approach\cite{chen2021phase,yangge}.

Recently, cell-free massive multiple-input multiple-output (MIMO) has emerged as a key enabling technology to achieve the 6G vision of ubiquitous connectivity, where all access points (APs) are connected to a central processing unit (CPU) via fronthaul links for joint signal processing\cite{lin2024communication}. This architecture eliminates traditional cell boundaries, thereby overcoming inter-cell interference. Meanwhile, by deploying the large-scale antenna arrays at the APs, cell-free massive MIMO can achieve exceptionally high system capacity and enhanced spatial resolution\cite{ganesan2021clustering}. 

On the other hand, as the array aperture enlarges, the boundary between the near-field and the far-field regions, characterized by the Rayleigh distance also expands~\cite{lu2024tutorial,cuilaoshi,wangzhe}. Consequently, devices may be randomly distributed in either the near-field or the far-field region of each AP. Due to the fundamental difference of the electromagnetic wave propagation between the near-field and the far-field channels, the conventional assumption of the far-field propagation alone becomes impractical, necessitating the consideration of the hybrid near-far field communications. Furthermore, due to the hybrid near-far field channels, the received signals at each AP exhibit  spatial correlations. Consequently, the rank-one update usually employed in the covariance-based method is not applicable anymore, rendering the algorithm design for activity detection mathematically intractable.

To overcome the above challenges, this paper proposes a practical hybrid near-far field system model for activity detection in cell-free massive MIMO. We derive the probability density function (PDF) of the received signals and establish a covariance-based formulation.  Furthermore, to reduce the computational burden at the CPU,  a distributed covariance-based activity detection algorithm is proposed to balance the computations across the network.  Specifically, each AP performs local activity detection for the devices and then transmits the detection result to the CPU for further
processing. In particular, a novel coordinate descent (CD) algorithm based on the Sherman-Morrison-Woodbury update with Taylor expansion is proposed to efficiently handle the local detection problem at each AP.

Moreover, we theoretically reveal that the hybrid near-far field channels can improve the detection performance compared with the conventional far-field channels. Simulation results demonstrate that the proposed approach achieves better detection performance compared to existing methods, and also corroborates the theoretical analysis.

\section{System Model and Problem Formulation}
\subsection{System Model}

As shown in~Fig~\ref{fig:Hybrid}, consider an uplink cell-free massive MIMO system with \( M \) APs and \( N \) IoT devices. Each AP is equipped with a $K$-antenna uniform linear array with a half-wavelength space, and each IoT device utilizes a single antenna. All \( M \) APs communicate with a CPU through fronthaul links. The channels between the APs and the devices are assumed to undergo quasi-static block fading, which  remains unchanged within each coherence block, but may vary across different blocks. Due to the boundary of the Rayleigh distance $2D^{2}/\lambda_\text{c}$, where $D$ is the aperture of the antenna array, and $\lambda_\text{c}$ is the wavelength of the central carrier, some devices are located in the near-field region of an AP if they are within its  Rayleigh distance, and the remaining devices are located in its far-field region. For each AP $m$, let $\mathcal{U}_m$ denote the subset containing  $N_{m,\text{near}}$  devices in the near-field region, and $\mathcal{U}^{\text{c}}_m$ denote the complement set containing the remaining $N - N_{m,\text{near}}$ devices in the far-field region. If  device $n$ is located in the far-field region of  AP $m$,  the uplink channel  can be modeled as \cite{ganesan2021clustering}:
\begin{equation}
\label{far-field chanel}
   \mathbf{h}_{m,n}\sim\mathcal{CN}(\mathbf{0},g_{m,n}\mathbf{I}_{K}),~\forall~n \in \mathcal{U}^{\text{c}}_m,
\end{equation}
 where $g_{m,n}$ is the large-scale fading coefficient. If  device $n$ is located in the near-field region of AP $m$, its uplink channel consists of a deterministic line-of-sight (LoS) channel and a statistical multi-path non-line-of-sight (NLoS) channel induced by $L_{m}$ scatters\cite{wang2024beamfocusing}. Then, the near-field channel can be modeled as~\cite[Eq.~(132)]{liuNearFieldCommunicationsTutorial2023}:
 \begin{align}
    \mathbf{h}_{m,n} \!\!=\! \underbrace{\beta_{m,n} \mathbf{b}_{m,n}(\mathbf{r}_{m,n})}_{\text{LoS}} \!+\!\! \sum_{\ell=1}^{L_{m}}\underbrace{ \varphi_{m,\ell}\tilde{\beta}_{m,n,\ell} \mathbf{b}_{m,n}(\tilde{\mathbf{r}}_{m,\ell})}_{\text{NLoS}}, \forall~n\in\mathcal{U}_m,
 \end{align}
\addtolength{\topmargin}{0.01in}
 where $\beta_{m,n}$ and $\tilde{\beta}_{m,n,\ell}$ denote the LoS and the NLoS channel gains, respectively,  $\varphi_{m,l}\sim\mathcal{CN}(0,\sigma_{m,l}^2)$ is the reflection coefficient of the $l$-th scatterer of AP $m$ with variance $\sigma_{m,l}^2$, $[\mathbf{r}_{m,n}]_k$ is the distance between device $n$ and the $k$-th antenna of  AP $m$, and $[\tilde{\mathbf{r}}_{m,\ell}]_k$ denotes the  distance between scatter $\ell$ and the $k$-th antenna of AP $m$. Furthermore, $\mathbf{b}_{m,n}(\cdot)$ denotes the near-field array response. For instance, for the LoS channel, the array response is given by
 \begin{align}
    \mathbf{b}_{m,n}(\mathbf{r}_{m,n}) =
    \big[ & e^{-j\frac{2\pi}{\lambda}([\mathbf{r}_{m,n}]_1- r_{m,n,0})}, 
          e^{-j\frac{2\pi}{\lambda}([\mathbf{r}_{m,n}]_2- r_{m,n,0})}, \notag \\
          & \dots, e^{-j\frac{2\pi}{\lambda}([\mathbf{r}_{m,n}]_K- r_{m,n,0})} \big]^{\text{T}},
\end{align}

\noindent where $r_{m,n,0}$ is the distance between device $n$ and the central element of AP $m$. Then, we can calculate the covariance matrix of $\mathbf{h}_{m,n}$ as
\begin{align}
    \mathbf{R}_{m,n} = \sum_{\ell=1}^{L_{m}}\sigma_{m,\ell}^2 |\tilde{\beta}_{m,n,\ell}|^2 \mathbf{b}_{m,n}(\tilde{\mathbf{r}}_{m,\ell})\mathbf{b}_{m,n}^{\text{H}}(\tilde{\mathbf{r}}_{m,\ell}),~\forall~n\in\mathcal{U}_m.
\end{align}
On this basis, the near-field channel $\mathbf{h}_{m,n}$ can be modeled as
\begin{align}
    \label{near-field channel}
    \mathbf{h}_{m,n} \sim \mathcal{CN}(\beta_{m,n} \mathbf{b}_{m,n}(\mathbf{r}_{m,n}), \mathbf{R}_{m,n}),~\forall~n\in\mathcal{U}_m.
\end{align}

\begin{figure}[t!]
    \centering
    \includegraphics[width=0.7\linewidth]{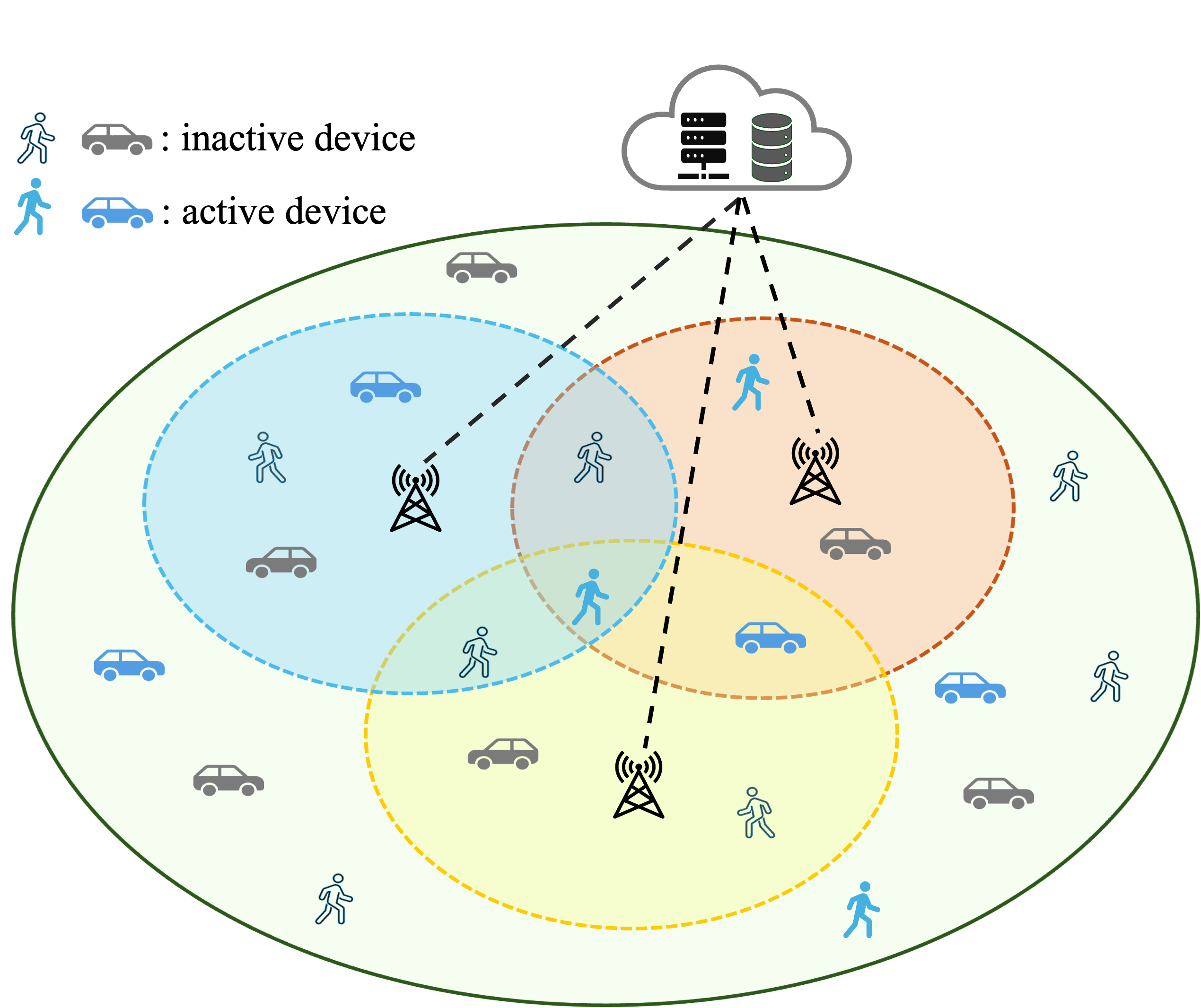}
    \caption{Hybrid near-far field activity detection in cell-free massive MIMO.}
    \label{fig:Hybrid}
\end{figure}

To detect the device activity, each device $n$ is pre-assigned a distinct signature sequence $\mathbf{s}_n \in \mathbb{C}^L,$  where $L$ represents the length of the signature sequence.  Given the sporadic nature of mMTC, only a small fraction of the  $N$  devices remain active within each coherence block. Let $a_n\in \{0,1\}$ denote the binary variable for device $n$ ($a_n= 1$ denotes the active status).  Then, the received signal at AP \( m \) can be written as
\begin{align}
\label{receive signal}
\mathbf{Y}_m & = \sum_{n \in \mathcal{U}_m} a_n \mathbf{s}_n \mathbf{h}_{m,n}^\text{T} + \sum_{n \in \mathcal{U}_m^{\text{c}}} a_n \mathbf{s}_n \mathbf{h}_{m,n}^\text{T} + \mathbf{W}_m,
\end{align}

\noindent where  \( \mathbf{W}_m\in \mathbb{C}^{L\times K}\) is the additive independent and identically distributed (i.i.d.) Gaussian noise at AP \( m \) with each entry following $\mathcal{CN}(0,\varsigma^2_m)$ and $\varsigma^2_m$ denoting the noise variance.

\subsection{Problem Formulation}
Mathematically, the problem of activity detection is equivalent to determining each $a_{n}\in \{0, 1\}$.
Notice that the channels  and the noise  are all complex Gaussian random variables, while $\{\mathbf{s}_n\}_{n=1}^{N}$ and $\{{a_n}\}_{n=1}^{N}$ can be viewed as deterministic. Consequently, the received signal $\mathbf{Y}_m$ in (\ref{receive signal}) follows a complex Gaussian distribution as well. Given that the columns of $\mathbf{Y}_m$  are  not i.i.d. due to the near-field correlated channels, we analyze the joint distribution of all columns collectively. Thus, we consider its vectorized form $\mathbf{y}_m = \text{vec}(\mathbf{Y}_m)$, which can be written as
\begin{align}
\mathbf{y}_m &= \sum_{n \in \mathcal{U}_m} a_n \mathbf{h}_{m,n} \otimes \mathbf{s}_n +\sum_{n \in \mathcal{U}_m^\text{c}} a_n \mathbf{h}_{m,n} \otimes \mathbf{s}_n + \mathbf{w}_m,
\end{align}

\noindent where $\otimes$ denotes the Kronecker product, and \( \mathbf{w}_m \) is the vectorized noise. Thus, we have $\mathbf{y}_m\sim\mathcal{CN}(\bar{\mathbf{y}}_{m},\mathbf{C}_{m})$, where the  mean and the covariance matrix are calculated based on (\ref{far-field chanel}) and (\ref{near-field channel}), which are given by
\begin{align}
\label{expectation and covariance}
    \bar{\mathbf{y}}_{m} &=  \sum_{n \in \mathcal{U}_m} a_n \left(\beta_{m,n} \mathbf{b}_{m,n}(\mathbf{r}_{m,n})\right) \otimes \mathbf{s}_n, \notag \\
    \mathbf{C}_{m} &= \sum_{n \in \mathcal{U}_m} a_n \mathbf{R}_{m,n} \otimes (\mathbf{s}_n \mathbf{s}_n^\text{H}) + \sum_{n \in \mathcal{U}_m^\text{c}} a_n g_{m,n} \mathbf{I}_K \otimes (\mathbf{s}_n \mathbf{s}_n^\text{H}) \notag \\
    &\quad\quad + \varsigma^2_m\mathbf{I}_{LK}.
\end{align}

\noindent  The joint PDF of received signals \( \{\mathbf{y}_m\}_{m=1}^M \)  is given by
\begin{align}
&p( \{\mathbf{y}_m\}_{m=1}^M ; \mathbf{a})\notag\\
 = &\prod_{m=1}^M \frac{1}{\left|\pi \mathbf{C}_{m}\right |} \exp\left(-(\mathbf{y}_m - \bar{\mathbf{y}}_{m})^\text{H} \mathbf{C}_{m}^{-1} (\mathbf{y}_m - \bar{\mathbf{y}}_{m})\right),
\end{align}
\noindent where $\mathbf{a} \triangleq [a_1, a_2, \ldots, a_N]^\text{T}$ denotes the activity status vector. Our goal is to minimize the negative log-likelihood function \( -\log p( \{\mathbf{y}_m\}_{m=1}^M ; \mathbf{a}) \), which is formulated as
\begin{align}
 \label{eq:optimization_problem}
\underset{\mathbf{a}\in [0, 1]^{N}}{\text{min}}\sum_{m=1}^M \left\{\log|\mathbf{C}_{m}| \!+ \!(\mathbf{y}_m \!-\!\bar{\mathbf{y}}_{m})^\text{H}\mathbf{C}_{m}^{-1} (\mathbf{y}_m \!- \!\bar{\mathbf{y}}_{m})\right\}.
\end{align}
Once problem (\ref{eq:optimization_problem}) is solved, the device activity is detected as \(\hat{a}_{n} = \mathbb{I}(a_{n} \geq \gamma)\), where $\mathbb{I}(\cdot)$ is an indicator function, and \(\gamma\) serves as a threshold within the range \([0, 1]\) to balance the probability of missed detection (PM) and the probability of false alarm (PF)~\cite{ganesan2021clustering}.

\section{Distributed  activity detection}
In this section, we develop a distributed algorithm to tackle problem (\ref{eq:optimization_problem}). Unlike the centralized method where the CPU handles all computations, the proposed distributed algorithm operates on both the APs and the CPU. Each AP conducts their local detection and then transmits the detection result to the CPU. Then, the CPU  aggregates these results and distributes the combined information back to all APs to prepare for the subsequent iteration.

To enable distributed processing, we first reformulate problem (\ref{eq:optimization_problem}) as an equivalent consensus form:
\begin{subequations}
\label{distributed formulation}
\begin{align}
&\min_{\left \{{\boldsymbol{\theta}_{m}}\right \}_{m=1}^{M}, \mathbf {a}\in \left [{0,1}\right]^{N}}\, 
\sum _{m=1}^{M}f_{m}\left ({\boldsymbol{\theta}_{m}}\right), 
\\&\qquad ~\,\text {s.t.}~\boldsymbol{\theta}_{m}=\mathbf {a},\quad \forall~m=1,2,\ldots,M,
\end{align}    
\end{subequations}
where  $f_{m}\left ({\boldsymbol{\theta}_{m}}\right)= \log|\tilde{\mathbf{C}}_{m}| + (\mathbf{y}_m - \tilde{\bar{\mathbf{y}}}_{m})^\text{H} \tilde{\mathbf{C}}_{m}^{-1} (\mathbf{y}_m - \tilde{\bar{\mathbf{y}}}_{m})$, and $\tilde{\mathbf{C}}_{m}$ and $ \tilde{\bar{\mathbf{y}}}_{m}$ follow the same structure as $\mathbf{C}_{m}$ and $\bar{\mathbf{y}}_{m}$ but with $\boldsymbol{\theta}_{m}$ replacing \(\mathbf{a}\) in (\ref{expectation and covariance}). Note that $f_m(\cdot)$ only depends on the local calculation at the $m$-th AP, making it suitable for distributed optimization.

To solve problem (\ref{distributed formulation}) in a distributed manner, we construct its augmented Lagrangian function:
\begin{align}
\label{lagrangian}
&\mathcal{L} \left({\left \{{\boldsymbol{\theta}_{m}}\right \}_{m=1}^{M}, \mathbf {a}; \left \{{\boldsymbol {\lambda }_{m}}\right \}_{m=1}^{M} }\right) \nonumber \\
=&\sum \limits _{m=1}^{M} \left\{f_{m}\left ({\boldsymbol{\theta}_{m}}\right)+\boldsymbol {\lambda }_{m}^{\mathrm {T}}\left ({\boldsymbol{\theta}_{m}-\mathbf {a}}\right) +\frac {\mu }{2}\left \Vert{ \boldsymbol{\theta}_{m}-\mathbf {a}}\right \Vert_{2}^{2}\right\},
\end{align}
where $\boldsymbol{\lambda}_m \in \mathbb{R}^{N}$ is the dual variable with the consensus constraint $\boldsymbol{\theta}_{m} = \mathbf{a}$, and $\mu > 0$ is a penalty parameter. We adopt the alternating direction method of multipliers framework by alternately updating $\mathbf{a}$, $\{\boldsymbol{\theta}_{m}\}_{m=1}^M$, and $\{\boldsymbol{\lambda}_m\}_{m=1}^M$.

\subsubsection{Subproblem with respect to $\left \{{\boldsymbol{\theta}_{m}}\right \}_{m=1}^{M}$} At the $i$-th iteration, we decompose the problem with respect to $\left \{{\boldsymbol{\theta}_{m}}\right \}_{m=1}^{M}$ into $M$ parallel \mbox{subproblems}, with each  processed by the corresponding AP and given by
\begin{align}
\label{sub b}
&\min_{\boldsymbol{\theta}_{m}}f_{m}\left (\!{\boldsymbol{\theta}_{m}}\!\right)\!+\!\left ({\boldsymbol {\lambda }_{m}^{(i-1)}}\right)^{\mathrm {T}} \left ({\boldsymbol{\theta}_{m}\!-\!\mathbf {a}^{(i-1)}}\right)\! +\!\frac {\mu }{2}\left \Vert{ \boldsymbol{\theta}_{m}\!-\!\mathbf {a}^{(i-1)}}\right \Vert_{2}^{2}. 
\end{align}
This is a single-cell  activity detection
problem with additional
linear and quadratic terms. Compared with the traditional far-field scenario, the dimension of the covariance matrix $\tilde{\mathbf{C}}_{m}$ in problem (\ref{sub b}) is \(LK \times LK\), which is much larger than that in far-field scenario, i.e., \(L \times L\) \cite{chen2021phase,lin2024communication,ganesan2021clustering}. Therefore, the complexity of updating and storing matrix $\tilde{\mathbf{C}}_{m}$ is much higher than that in far-field scenario. Moreover, due to existence of channel correlation, the rank-one update in existing covariance-based approach \cite{ganesan2021clustering,li2022asynchronous,wang2024scalinglaw} is not applicable anymore, making problem (\ref{sub b}) challenging to solve.

To tackle the above challenges, we propose a novel CD algorithm based
on the Sherman-Morrison-Woodbury update with Taylor expansion.  For notational simplicity, we first introduce several unified symbols. Let 
\begin{align}
\bar{\mathbf{h}}_{m,n} &= 
\begin{cases} 
\beta_{m,n} \mathbf{b}_{m,n}(\mathbf{r}_{m,n}), & \text{if } n \in \mathcal{U}_m, \\
\mathbf{0}, & \text{if } n \in \mathcal{U}_m^{\text{c}},
\end{cases} \\
\mathbf{\Xi}_{m,n} &=
\begin{cases} 
\mathbf{R}_{m,n}, & \text{if } n \in \mathcal{U}_m, \\
g_{m,n}\mathbf{I}_K, & \text{if } n \in \mathcal{U}_m^{\text{c}}.
\end{cases}
\end{align}
Then, we can define a matrix $\mathbf{X}_{m,n}=\mathbf{\Xi}_{m,n}^{\frac{1}{2}} \otimes \mathbf{s}_n$ so that
\begin{align}
\label{defineX_mn}
 \mathbf{X}_{m,n} \mathbf{X}_{m,n}^\text{H} = \mathbf{\Xi}_{m,n} \otimes \left( \mathbf{s}_n \mathbf{s}_n^\text{H} \right).
\end{align}
 Define \(\mathbf{e}_n \in \mathbb{R}^N\) as the standard basis vector whose \(n\)-th entry is 1 while all other entries are 0.  For the $n$-th coordinate of $\boldsymbol{\theta}_m$, the update of $\theta_{m,n}$ can be expressed as $\theta_{m,n}+d$ with $\{\theta_{m,n'}\}_{n'=1,n'\neq n}^N$ fixed, where $d$ is determined by solving the following one-dimensional subproblem:
\begin{align}
\label{origin_update_distributed}
  & \underset{d \in [-\theta_{m,n},\,1-\theta_{m,n}]}{\operatorname{min}} \log \left| \tilde{\mathbf{C}}_{m} + d \mathbf{X}_{m,n} \mathbf{X}_{m,n}^{\text{H}} \right|  \notag\\
    &\quad\quad+ \left( \mathbf{y}_m - \tilde{\bar{\mathbf{y}}}_{m} - d\bar{\mathbf{h}}_{m,n} \otimes \mathbf{s}_n \right)^{\text{H}} \left( \tilde{\mathbf{C}}_{m} + d \mathbf{X}_{m,n} \mathbf{X}_{m,n}^{\text{H}} \right)^{-1}\notag\\
    &\quad\quad\times  \left( \mathbf{y}_m - \tilde{\bar{\mathbf{y}}}_{m} - d\bar{\mathbf{h}}_{m,n} \otimes \mathbf{s}_n \right) \notag\\
    &\quad\quad+ (\boldsymbol{\lambda}_m^{(i-1)})^{\text{T}}(\boldsymbol{\theta}_{m} \!\!+\! d\mathbf{e}_n \!\!-\! \mathbf{a}^{(i-1)})
   \! +\! \frac{\mu}{2}\|\boldsymbol{\theta}_{m}\! + \!d\mathbf{e}_n \!\!- \!\mathbf{a}^{(i-1)}\|_2^2. \notag\\
\end{align}

\noindent According to the property of the determinant, we have
\begin{align}\label{eq:term1}
	& \log\left| \tilde{\mathbf{C}}_{m} + d \, \mathbf{X}_{m,n} \mathbf{X}_{m,n} ^\text{H}\right| \nonumber \\
	= &\log \left| \mathbf{I}_{LK} + d\,\mathbf{X}_{m,n} \mathbf{X}_{m,n} ^\text{H} \tilde{\mathbf{C}}_{m}^{-1} \right| + \log \left| \tilde{\mathbf{C}}_{m}\right| \nonumber \\
	= &\log \left| \mathbf{I}_{J_{m,n}} + d\, \mathbf{X}_{m,n} ^\text{H} \tilde{\mathbf{C}}_{m}^{-1} \mathbf{X}_{m,n} \right| + \log \left| \tilde{\mathbf{C}}_{m}\right|\nonumber\\
    \overset{(a)}{\approx} &d\operatorname{tr} \left( \mathbf{X}_{m,n} ^\text{H} \tilde{\mathbf{C}}_{m}^{-1} \mathbf{X}_{m,n} \right) + \log \left| \tilde{\mathbf{C}}_{m}\right|,
\end{align}
where $J_{m,n}= \operatorname{rank}(\mathbf{\Xi}_{m,n})$, and $(a)$ is derived from the first-order Taylor expansion to avoid the computational complexity in calculating the determinant. On the other hand, based on the Sherman-Morrison-Woodbury formula, we have
\begin{align}\label{eq:term2}
	&\left(\tilde{\mathbf{C}}_{m} + d \mathbf{X}_{m,n} \mathbf{X}_{m,n}^{\text{H}} \right)^{-1} \notag\\
    =&\,\,\,\tilde{\mathbf{C}}_{m}^{-1}- d\, \tilde{\mathbf{C}}_{m}^{-1} \mathbf{X}_{m,n}\notag\\
    &\quad\quad\,\,\,\times\left( \mathbf{I}_{J_{m,n}} + d\, \mathbf{X}_{m,n}^{\text{H}} \tilde{\mathbf{C}}_{m}^{-1} \mathbf{X}_{m,n} \right)^{-1} \mathbf{X}_{m,n}^{\text{H}} \tilde{\mathbf{C}}_{m}^{-1}\notag\\
	\overset{(b)}{\approx}&\,\,\,\tilde{\mathbf{C}}_{m}^{-1}- d\, \tilde{\mathbf{C}}_{m}^{-1} \mathbf{X}_{m,n} \left( \mathbf{I}_{J_{m,n}} - d\, \mathbf{X}_{m,n}^{\text{H}} \tilde{\mathbf{C}}_{m}^{-1} \mathbf{X}_{m,n} \right) \notag\\
    &\quad\quad\,\,\,\times\mathbf{X}_{m,n}^{\text{H}} \tilde{\mathbf{C}}_{m}^{-1},
		\end{align}
where $(b)$ also applies the first-order Taylor expansion to avoid the computational complexity in repeatedly calculating the matrix inverse. Substituting the right-hand sides of (\ref{eq:term1}) and (\ref{eq:term2}) into (\ref{origin_update_distributed}), we obtain the CD update as shown in (\ref{eq:distributed inexact}) at the top of the next page, where $\omega \ge 0$ is a properly selected parameter to control the convergence of the algorithm, and $\lambda_{m,n}^{(i-1)}$ denotes the $n$-th entry of $\boldsymbol {\lambda }_{m}^{(i-1)}$. Setting the gradient of  (\ref{eq:distributed inexact}) to zero, the roots can be calculated by the cubic formula with the complexity of $\mathcal{O}(1)$~\cite{wang2024covariance}. The optimal solution of problem (\ref{eq:distributed inexact})  can be obtained by selecting the one with the smallest objective value among the roots. After sequentially updating $\theta_{m,n}$ for all $n$, problem (\ref{eq:distributed inexact}) can be efficiently solved to obtain $\boldsymbol{\theta}^{(i)}_m$ at each AP $m$.

\begin{figure*}
	\vspace{-6pt}
\begin{align}\label{eq:distributed inexact}
&\underset{d \in [-\theta_{m,n},\,1-\theta_{m,n}]}{\operatorname{min}}  d \rho_{1}(\boldsymbol{\theta}_{m}) + d^2 \rho_{2}(\boldsymbol{\theta}_{m}) + d^3 \rho_{3}(\boldsymbol{\theta}_{m}) + d^4 \rho_{4}(\boldsymbol{\theta}_{m}) + \frac{\omega}{2} d^2\nonumber\\
&\quad\quad\quad\rho_{1}(\boldsymbol{\theta}_{m}) = \text{tr} \left( \mathbf{X}_{m,n}^\text{H} \tilde{\mathbf{C}}_{m}^{-1} \mathbf{X}_{m,n} \right) - 2 \, \text{Re} \left( \left( \mathbf{y}_{m} - \tilde{\bar{\mathbf{y}}}_{m} \right)^\text{H} \tilde{\mathbf{C}}_{m}^{-1} \left( \bar{\mathbf{h}}_{m,n} \otimes \mathbf{s}_n \right) \right) \nonumber\\
&\quad\quad\quad\quad\quad\quad\quad - \left( \mathbf{y}_{m} - \tilde{\bar{\mathbf{y}}}_{m} \right)^\text{H} \tilde{\mathbf{C}}_{m}^{-1} \mathbf{X}_{m,n} \mathbf{X}_{m,n}^\text{H} \tilde{\mathbf{C}}_{m}^{-1} \left( \mathbf{y}_{m} - \tilde{\bar{\mathbf{y}}}_{m} \right) + \lambda_{m, n}^{(i-1)} + \mu \left(\theta_{m,n}-a_{n}^{(i-1)}\right), \nonumber\\
&\quad\quad\quad\rho_{2}(\boldsymbol{\theta}_{m}) = \left( \bar{\mathbf{h}}_{m,n} \otimes \mathbf{s}_n \right)^\text{H} \tilde{\mathbf{C}}_{m}^{-1} \left( \bar{\mathbf{h}}_{m,n} \otimes \mathbf{s}_n \right) + 2 \, \text{Re} \left( \left( \mathbf{y}_{m} - \tilde{\bar{\mathbf{y}}}_{m} \right)^\text{H} \tilde{\mathbf{C}}_{m}^{-1} \mathbf{X}_{m,n} \mathbf{X}_{m,n}^\text{H} \tilde{\mathbf{C}}_{m}^{-1} \left( \bar{\mathbf{h}}_{m,n} \otimes \mathbf{s}_n \right) \right) \nonumber\\
&\quad\quad\quad\quad\quad\quad\quad + \left( \mathbf{y}_{m} - \tilde{\bar{\mathbf{y}}}_{m} \right)^\text{H} \tilde{\mathbf{C}}_{m}^{-1} \mathbf{X}_{m,n} \mathbf{X}_{m,n}^\text{H} \tilde{\mathbf{C}}_{m}^{-1} \mathbf{X}_{m,n} \mathbf{X}_{m,n}^\text{H} \tilde{\mathbf{C}}_{m}^{-1} \left( \mathbf{y}_{m} - \tilde{\bar{\mathbf{y}}}_{m} \right) + \frac{\mu}{2}, \nonumber\\
&\quad\quad\quad\rho_{3}(\boldsymbol{\theta}_{m}) = -2 \, \text{Re} \left( \left( \mathbf{y}_{m} - \tilde{\bar{\mathbf{y}}}_{m} \right)^\text{H} \tilde{\mathbf{C}}_{m}^{-1} \mathbf{X}_{m,n} \mathbf{X}_{m,n}^\text{H} \tilde{\mathbf{C}}_{m}^{-1} \mathbf{X}_{m,n} \mathbf{X}_{m,n}^\text{H} \tilde{\mathbf{C}}_{m}^{-1} \left( \bar{\mathbf{h}}_{m,n} \otimes \mathbf{s}_n \right) \right) \nonumber\\
&\quad\quad\quad\quad\quad\quad\quad - \left( \bar{\mathbf{h}}_{m,n} \otimes \mathbf{s}_n \right)^\text{H} \tilde{\mathbf{C}}_{m}^{-1} \mathbf{X}_{m,n} \mathbf{X}_{m,n}^\text{H} \tilde{\mathbf{C}}_{m}^{-1} \left( \bar{\mathbf{h}}_{m,n} \otimes \mathbf{s}_n \right), \nonumber\\
&\quad\quad\quad\rho_{4}(\boldsymbol{\theta}_{m}) = \left( \bar{\mathbf{h}}_{m,n} \otimes \mathbf{s}_n \right)^\text{H} \tilde{\mathbf{C}}_{m}^{-1} \mathbf{X}_{m,n} \mathbf{X}_{m,n}^\text{H} \tilde{\mathbf{C}}_{m}^{-1} \mathbf{X}_{m,n} \mathbf{X}_{m,n}^\text{H} \tilde{\mathbf{C}}_{m}^{-1} \left( \bar{\mathbf{h}}_{m,n} \otimes \mathbf{s}_n \right).
\end{align}
\vspace{-5pt}
\hrulefill
\end{figure*}

\subsubsection{Updating the dual variables $\{\boldsymbol{\lambda}_{m}\}_{m=1}^{M}$} After updating  $\{\boldsymbol{\theta}_{m}\}_{m=1}^M$, $\{\boldsymbol{\lambda}_{m}\}_{m=1}^{M}$ are updated by a dual ascent step:
\begin{align}
\label{dual ascent step}
\boldsymbol{\lambda}_m^{(i)} = \boldsymbol{\lambda}_m^{(i-1)} \!+\! \mu \left( \boldsymbol{\theta}_{m}^{(i)} - \mathbf{a}^{(i-1)} \right), \forall~m = 1,2, \dots, M. 
\end{align}

\subsubsection{Subproblem with respect to $\mathbf{a}$}  We decompose the problem with respect to $\mathbf{a}$ into $N$ parallel subproblems, with each written as
\begin{align}
\label{sub a}
\min \limits _{a_n \in \left [{0,1}\right]} \sum _{m=1}^{M} \left\{{\lambda_{m,n}^{(i)}}\left ({\theta_{m,n}^{(i)}\!-\!a_n}\right)+\frac {\mu }{2}\left ({\theta_{m,n}^{(i)}\!-\!a_n}\right)^{2}\right\},
\end{align}
where  $\theta_{m,n}^{(i)}$ denotes the $n$-th entry of $\boldsymbol{\theta}_{m}^{(i)}$. Problem (\ref{sub a}) is a one-dimensional convex quadratic problem, and hence its optimal solution can be derived in a closed form:
\begin{align}
a_{n}^{(i)} = \Pi_{[0,1]} \left( \delta_{n}^{(i)} \right), 
\quad \forall~n = 1,2, \dots, N,
\label{eq:close form}
\end{align}
where $\delta_{n}^{(i)} = \sum_{m=1}^M \left( \mu \theta_{m,n}^{(i)} + \lambda_{m,n}^{(i)}\right)/(M\mu)$, and $\Pi_{[0,1]} (\cdot)$ is the projection operation onto $[0,1]$.

\addtolength{\topmargin}{0.05in}
\begin{algorithm}[H]
\caption{Proposed Distributed Algorithm for Solving Problem (\ref{distributed formulation})}
\label{distributed algorithm}
\begin{algorithmic}[1]
\STATE \textbf{Initialize:}  $ \boldsymbol{\theta}_{m}^{(0)}$, $\boldsymbol{\lambda}_{m}^{(0)}$,  $\forall~m = 1, 2, \dots, M$, and $\mathbf{a}^{(0)}$ using~(\ref{eq:close form});
\STATE \textbf{repeat} ($i = 1, 2,\dots$)
\STATE \quad The CPU broadcasts $\mathbf{a}^{(i-1)}$ to each AP $m$, $\forall~m = 1, 2,\dots, M$;
\STATE  \quad Each AP $m$ updates $\boldsymbol{\theta}_{m}^{(i)}$ by sequentially solving problem (\ref{eq:distributed inexact}) with respect to all coordinates, $\forall~m = 1,2,\dots, M$;
\STATE \quad Each AP $m$ updates $\boldsymbol{\lambda}_m^{(i)}$ by a dual ascent step~(\ref{dual ascent step}), $\forall~m = 1,2, \dots, M$;
\STATE \quad Each AP $m$ sends $\mu \boldsymbol{\theta}_{m}^{(i)} + \boldsymbol{\lambda}_m^{(i)}$ to the CPU, $\forall~m = 1,2, \dots, M$;
\STATE \quad The CPU updates $\mathbf{a}^{(i)}$ with (\ref{eq:close form});
\STATE \textbf{until} convergence
\end{algorithmic}
\end{algorithm}

Through iterative updates of both primal and dual variables, we present a distributed algorithm to solve problem (\ref{distributed formulation}), with the complete procedure outlined in Algorithm~\ref{distributed algorithm}. Although problem (\ref{distributed formulation}) is nonconvex, the following theorem demonstrates that Algorithm~\ref{distributed algorithm} can converge to a stationary point, We omit the proof due to space limitations.

\begin{theorem}\label{Distributed Convergency Theorem}

Let $\tilde{L}_m$ and $L_m$ denote the Lipschitz constant of $\nabla f_m(\boldsymbol{\theta}_{m})$ and $\nabla U_m(\boldsymbol{\theta}_{m})$, respectively, where  $U_m(\boldsymbol{\theta}_m)$ denotes the objective function of problem (\ref{sub b}). There exists a constant $\rho=\bar{\rho}_2+\bar{\rho}_3+\bar{\rho}_4$ such that $\bar{\rho}_2$, $\bar{\rho}_3$, and $\bar{\rho}_4$ are the upper bounds of $\rho_{2}(\boldsymbol{\theta}_{m})$, $\rho_{3}(\boldsymbol{\theta}_{m})$, and $\rho_{4}(\boldsymbol{\theta}_{m})$ in (\ref{eq:distributed inexact}), respectively. When $\omega \ge L_{m} + \rho$, $\omega+\rho \ge 1$, and $\mu>2\tilde{L}_{m}$, the solution sequence  $\{\mathbf{a}^{(i)}\}$ generated by Algorithm \ref{distributed algorithm}  converges to a stationary point of problem (\ref{distributed formulation}).
\end{theorem}

\section{Performance Analysis}
In this section, we analyze how the hybrid near-far field channels affect the detection performance. To achieve good detection performance, the true activity indicator vector $\mathbf{a}^{\circ}$ should be uniquely identifiable. Specifically, there should not exist another vector $\tilde{\mathbf{a}} \neq \mathbf{a}^{\circ}$ that yields the same PDF, i.e., $p(\{\mathbf{y}_m\}_{m=1}^M ; \tilde{\mathbf{a}}) = p(\{\mathbf{y}_m\}_{m=1}^M ; \mathbf{a}^{\circ}).$ Given that both distributions are multivariate Gaussian, this equivalence requires identical means and covariance matrices.  According to the second equation of (\ref{expectation and covariance}), when the covariance matrices are identical, we have 
\begin{equation}
    \sum_{n \in \mathcal{U}_m} \!(\tilde{a}_n\!-\!a_n^{\circ}) \mathbf{R}_{m,n} \otimes (\mathbf{s}_n \mathbf{s}_n^\text{H})+\!\!\!\sum_{n \in \mathcal{U}_m^\text{c}}\!(\tilde{a}_n\!-\!a_n^{\circ}) g_{m,n} \mathbf{I}_K \otimes (\mathbf{s}_n \mathbf{s}_n^\text{H})\!=\!\mathbf{O},
\end{equation}
where $\tilde{a}_n$ and $a_n^{\circ}$ denote the $n$-th entry of $\tilde{\mathbf{a}}$ and $\mathbf{a}^{\circ}$, respectively, and  $\mathbf{O}$ denotes the zero matrix. After vectorization and using the definition of $\mathbf{X}_{m,n}$ in (\ref{defineX_mn}), we have
\begin{equation}
 \boldsymbol{\Psi}_m\boldsymbol{\xi}\! =\! \mathbf{0},
\end{equation}
where $\boldsymbol{\xi} = \tilde{\mathbf{a}} - \mathbf{a}^{\circ}$ and
\begin{align}
    \boldsymbol{\Psi}_m\!\!= \![\text{vec}(\mathbf{X}_{m,1}\mathbf{X}_{m,1}^\text{H}\!),\text{vec}(\mathbf{X}_{m,2}\mathbf{X}_{m,2}^\text{H}),...,\!\!\text{vec}(\mathbf{X}_{m,N}\mathbf{X}_{m,N}^\text{H})].
\end{align}
Let
\begin{align}
    \mathcal{V} &= \{\boldsymbol{\xi} \in \mathbb{R}^N | \boldsymbol{\Psi}_m\boldsymbol{\xi} = \mathbf{0}, \boldsymbol{\xi} \neq \mathbf{0}, \forall~m =1,2,...,M\}, \label{identify 1}\\
    \mathcal{D} &= \{\boldsymbol{\xi} \in \mathbb{R}^N | \xi_n \geq 0 \text{ if } a_n^{\circ} = 0, \xi_n \leq 0 \text{ if } a_n^{\circ} = 1\}\label{feasible set},
\end{align}
where  (\ref{identify 1}) defines the solution set containing any other $\tilde{\mathbf{a}}\neq \mathbf{a}^{\circ}$ that yields the same covariance matrix, while (\ref{feasible set}) defines the feasible set of problem (\ref{eq:optimization_problem}). Consequently, the identifiability of $\mathbf{a}^{\circ}$  requires at least \(\mathcal{V}\cap \mathcal{D} = \emptyset \). From (\ref{identify 1}), we can observe that  if the columns of  $\boldsymbol{\Psi}_m$ are closer to orthogonality, then the condition $\mathcal{V} \cap \mathcal{D} = \emptyset$ is more likely to hold, which ensures better detection performance. The following proposition demonstrates that the columns of  $\boldsymbol{\Psi}_m$  in the hybrid near-far field channels are closer to orthogonality than those of the far-field channels.

\begin{proposition}\label{prop:mutual_similarity}
For any two devices $n$ and $n'$, the columns of $\boldsymbol{\Psi}_m$ in the hybrid near-far field channels are closer to \mbox{orthogonality} in the sense that their cosine similarity is upper bounded by that of the far-field channels:
\begin{align}
\frac{(\boldsymbol{\psi}_{m,n})^{\text{H}} \boldsymbol{\psi}_{m,n'}}{\|\boldsymbol{\psi}_{m,n}\|_2 \|\boldsymbol{\psi}_{m,n'}\|_2} \leq \left(\frac{|\mathbf{s}_n^{\text{H}} \mathbf{s}_{n'}|}{\|\mathbf{s}_n\|_2 \|\mathbf{s}_{n'}\|_2}\right)^2,
\end{align}
where $\boldsymbol{\psi}_{m,n}$ is the $n$-th column of $\boldsymbol{\Psi}_m$, and  the equality holds if and only if $\mathbf{\Xi}_{m,n} = c\mathbf{\Xi}_{m,n'}$ for some constant $c$.
\end{proposition}
\begin{proof}
Please see Appendix~\ref{app:mutual_similarity}.
\end{proof}
The above analysis demonstrates that the hybrid near-far field channels can effectively improve the detection performance as compared to the conventional far-field channels. 

\section{Simulation results}
In this section, we validate the performance of the proposed method through simulations in terms of the PM and the PF~\cite{ganesan2021clustering}. We consider a \(200 \times 200\) square meters area, i.e., $[-100,100] \times [-100, 100]$, with wrap-around at the boundary. There are \(N = 100\) IoT devices uniformly distributed in this square area, where the ratio of the active devices is 0.1. The signature sequence of each device \(\mathbf{s}_n\) is generated randomly and uniformly from the discrete set $\left\{ \pm \frac{\sqrt{2}}{2} \pm j \frac{\sqrt{2}}{2} \right\}^L$ \cite{wang2024scalinglaw}. The number of the APs is set to $M = 3$, which are located at $(40,0)$, $(-20,20\sqrt{3})$, and $(-20,-20\sqrt{3})$, respectively.
 The  channel coefficients $g_{m,n}$, $\beta_{m,n}^2$, and $\tilde{\beta}^2_{m,n,\ell}$ follow the path-loss model, which is given by $128.1 + 37.6 \log_{10}(\tau),$ where $\tau$ is the corresponding distance in km, the number of scatters at each AP is set to $L_m=8$, and the variance of  the reflection coefficient of each scatter is set to $\sigma_{m,l}^2=1$.
The background noise power is set as $-99$ dBm.

\addtolength{\textheight}{-0.05in}

\begin{figure}[t]
	\centering
    	\begin{subfigure}[t]{0.49\columnwidth}
		\includegraphics[width=\textwidth]{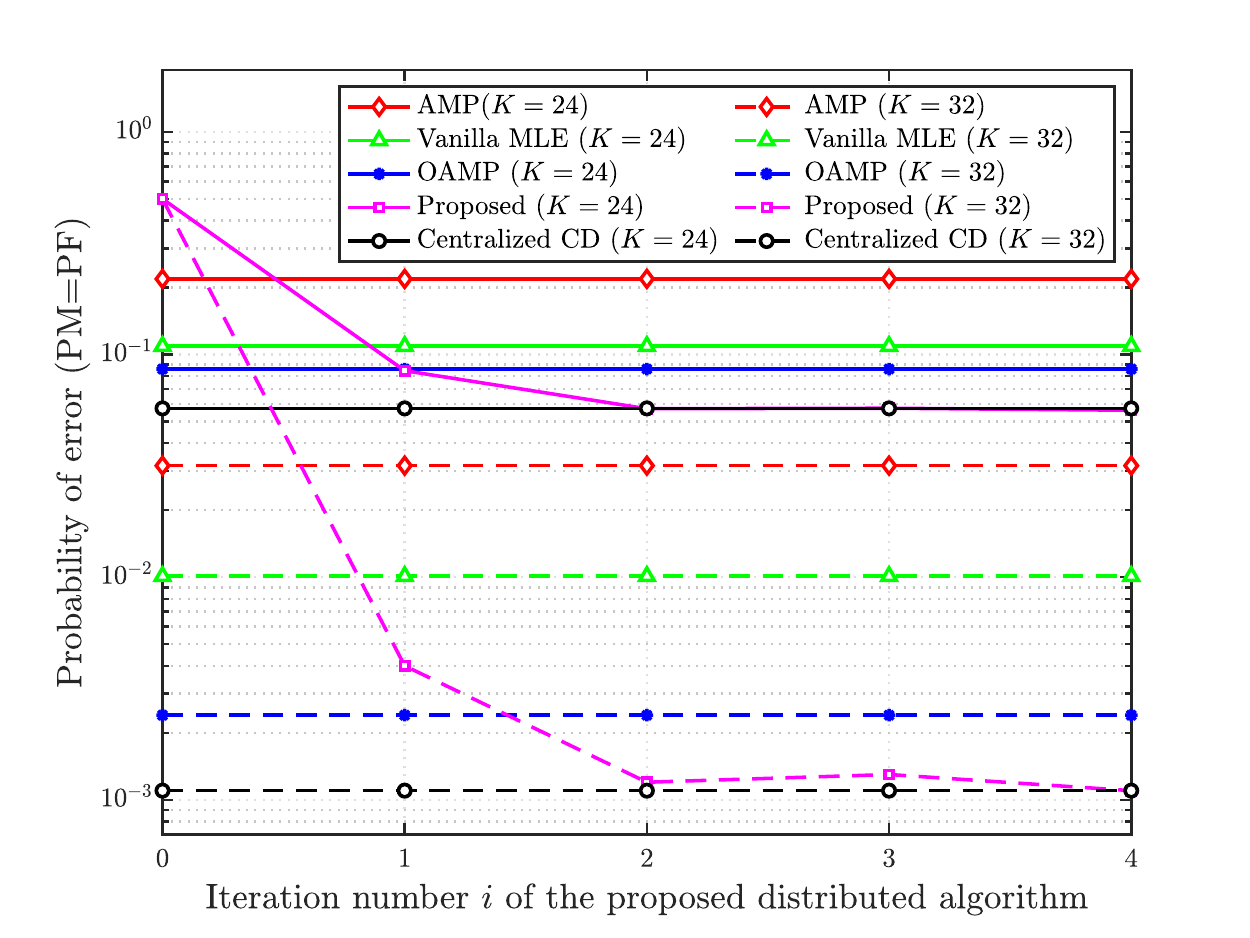}
		\caption{Probability of error versus iteration number.}
		\label{fig iteration}
	\end{subfigure}
	\begin{subfigure}[t]{0.49\columnwidth}
		\includegraphics[width=\textwidth]{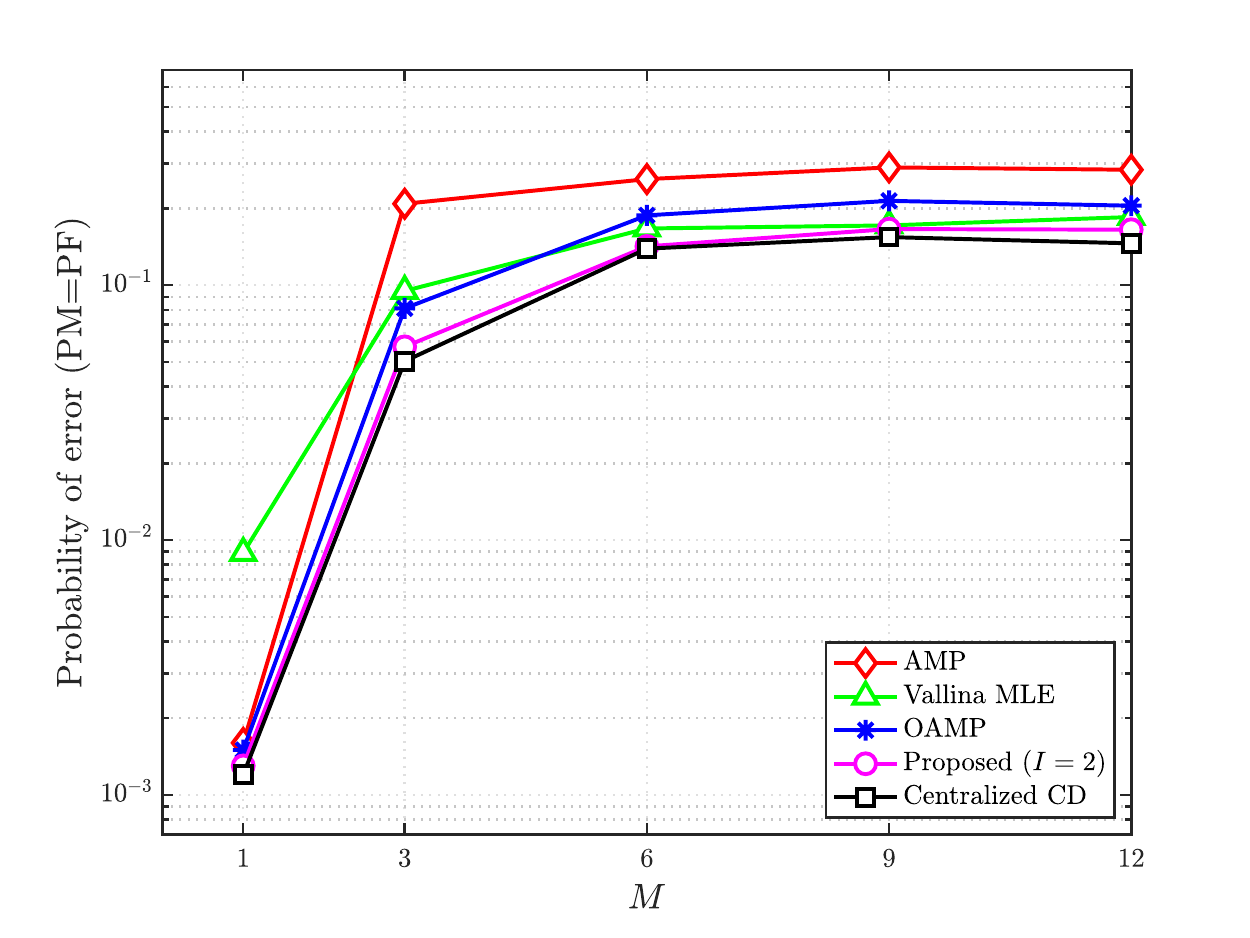}
		\caption{Probability of error versus $M$.}
		\label{prob of error versus number of APs}
	\end{subfigure}
	\caption{Comparison of the proposed distributed algorithm with the benchmarks.}
	\label{performance comparison}
	\vspace{-5pt}
\end{figure}


We compare the performance of the proposed distributed algorithm with the following benchmarks: Centralized CD~\cite{wang2024covariance}, AMP\cite{djelouat2021user}, OAMP\cite{cheng2020orthogonal}, and  Vanilla MLE\cite{liu2024mle}. The probability of error when $\text{PM}=\text{PF}$ is shown in Fig.~\ref{fig iteration}, where $L = 6$, $K=24$ or $32$,  and $D=(K-1)\lambda_\text{c}/2$ with $\lambda_{\text{c}}=0.2$ m.
 We can see that the proposed distributed  Algorithm~\ref{distributed algorithm} achieves a fast convergence speed within only 2 iterations, and the increased  number of the antennas can improve the detection performance of all methods. Moreover, since the proposed algorithm leverages the problem formulation (\ref{eq:optimization_problem}) based on an
 accurate hybrid-field channel model (\ref{expectation and covariance}), it outperforms the exiting centralized approaches, including AMP, OAMP, and Vanilla MLE. Finally, as a distributed algorithm, the proposed Algorithm~1 can achieve comparable performance to Centralized CD. Figure~\ref{prob of error versus number of APs} further demonstrates that when the total number of antennas is fixed as $72$ with signature sequence length $L=6$, the algorithm maintains superior performance as the number of APs varies from $1$ to $12$ with antennas distributed equally among the APs.

 \begin{figure}[t]
	\centering
    	\begin{subfigure}[t]{0.49\columnwidth}
		\includegraphics[width=\textwidth]{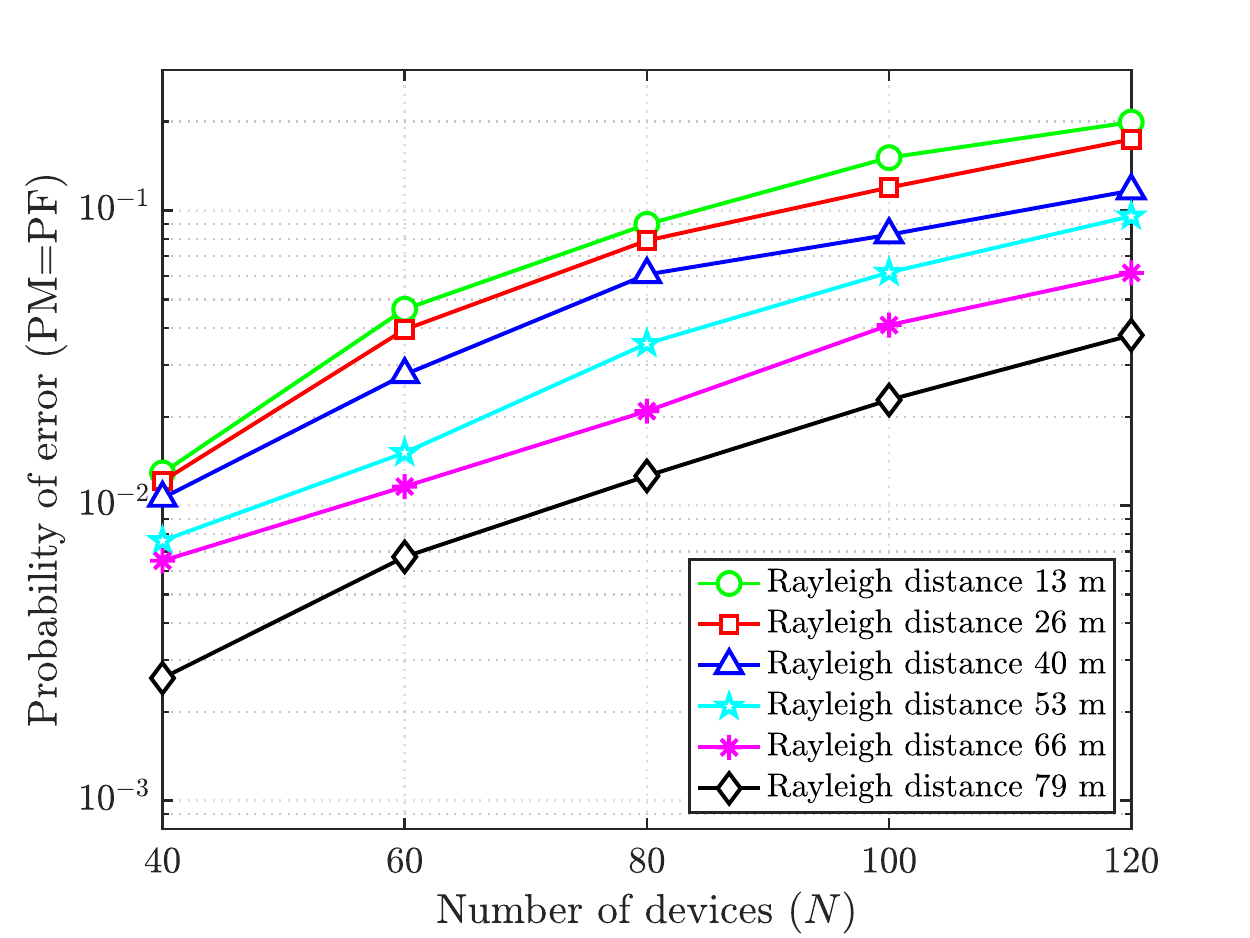}
		\caption{Probability of error versus $N$ when $L=6$.}
		\label{subfig-b}
	\end{subfigure}
	\begin{subfigure}[t]{0.49\columnwidth}
		\includegraphics[width=\textwidth]{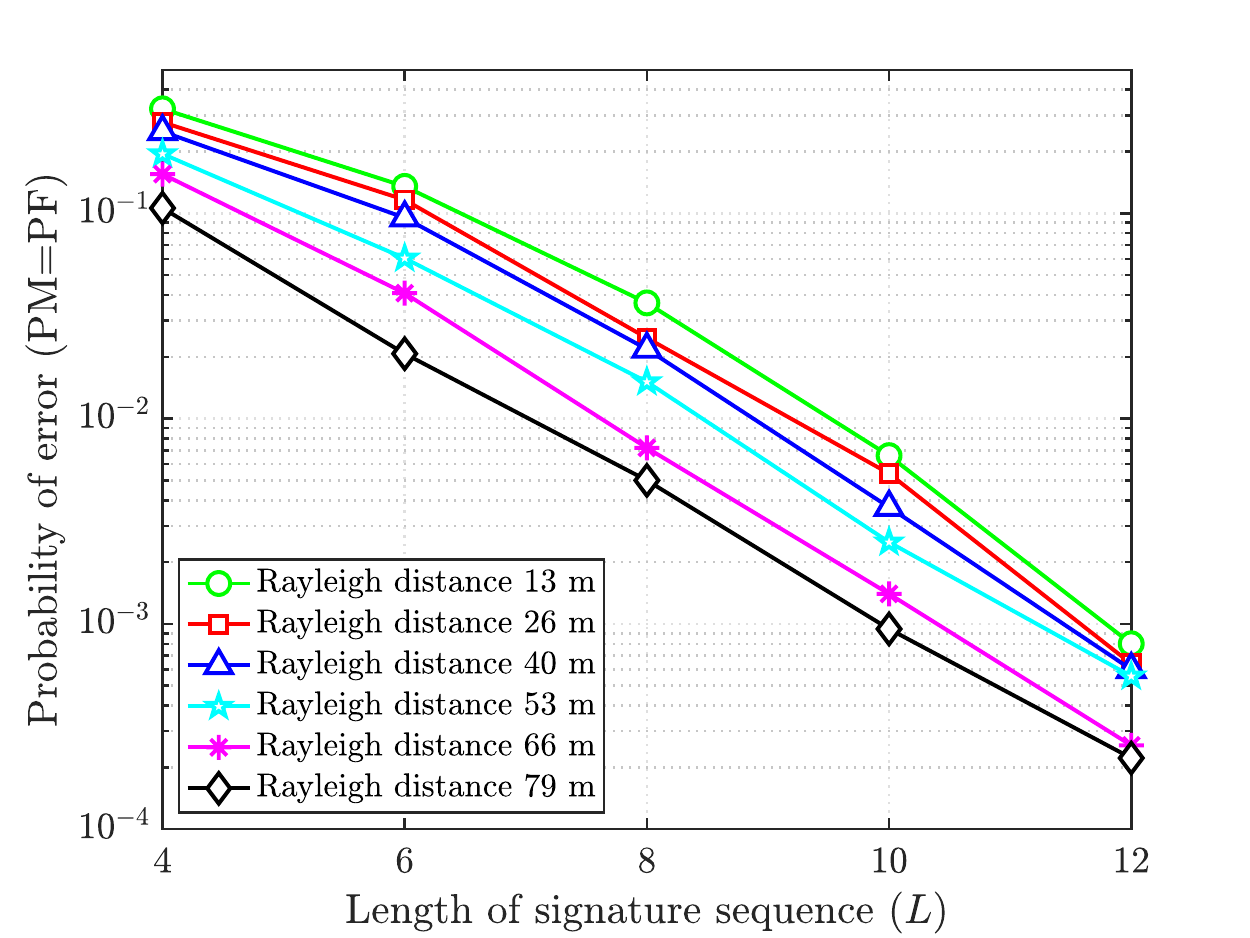}
		\caption{Probability of error versus $L$ when $N=100$.}
		\label{subfig-a}
	\end{subfigure}
	\caption{Performance comparison across different Rayleigh distances.}
	\label{analysis simulation}
	\vspace{-5pt}
\end{figure}

Figure~\ref{analysis simulation} illustrates the detection performance across different Rayleigh distances, where $K=24$, and $D=(K-1)\lambda_\text{c}/2$ with $\lambda_{\text{c}}=0.05$~m, $0.1$ m, $0.15$ m, $0.20$, $0.25$, and $0.3$ m, which correspond to the Rayleigh distances of $13$ m, $26$ m, $40$ m, $53$ m, $66$~m, and $79$~m, respectively. It can be observed that a longer Rayleigh distance, which incorporates more near-field channels, leads to better detection performance under different device numbers and different lengths of signature sequences, verifying the theoretical analysis in Proposition~1.

\section{Conclusion}
This paper investigated  the  hybrid near-far field activity detection in cell-free massive MIMO and established a covariance-based formulation. Then, a distributed activity detection algorithm was proposed to alleviate the computational burden at the CPU. The theoretical analysis demonstrated that the hybrid near-far field channels can improve the detection performance compared with the conventional far-field channels. Simulation results  corroborated the
theoretical analysis and demonstrated that the proposed distributed approach achieves superior detection performance compared to the existing methods.

\appendices

\section{Proof of Proposition~\ref{prop:mutual_similarity}}\label{app:mutual_similarity}
For devices $n$ and $n'$, we have
\begin{align}
\label{prove proposition1}
    \frac{(\boldsymbol{\psi}_{m,n})^{\text{H}} \boldsymbol{\psi}_{m,n'}}{\|\boldsymbol{\psi}_{m,n}\|_2 \|\boldsymbol{\psi}_{m,n'}\|_2}&\! =\! \frac{\text{vec}(\mathbf{X}_{m,n}\mathbf{X}_{m,n}^{\text{H}})^{\text{H}}\text{vec}(\mathbf{X}_{m,n'}\mathbf{X}_{m,n'}^{\text{H}})}{\|\text{vec}(\mathbf{X}_{m,n}\mathbf{X}_{m,n}^{\text{H}})\|_2 \|\text{vec}(\mathbf{X}_{m,n'}\mathbf{X}_{m,n'}^{\text{H}})\|_2} \notag\\
    &= \frac{\text{tr}(\mathbf{\Xi}_{m,n} \mathbf{\Xi}_{m,n'})}{\|\mathbf{\Xi}_{m,n}\|_F \|\mathbf{\Xi}_{m,n'}\|_F} \times \left(\frac{|\mathbf{s}_n^{\text{H}} \mathbf{s}_{n'}|}{\|\mathbf{s}_n\|_2 \|\mathbf{s}_{n'}\|_2}\right)^2 \notag\\
    & \overset{(a)}{\leq} \left(\frac{|\mathbf{s}_n^{\text{H}} \mathbf{s}_{n'}|}{\|\mathbf{s}_n\|_2 \|\mathbf{s}_{n'}\|_2}\right)^2,
\end{align}
where  $(a)$ follows from the Cauchy-Schwartz inequality with the equality holding if and only if $\mathbf{\Xi}_{m,n} = c\mathbf{\Xi}_{m,n'}$ for some constant $c$. The right-hand side of (\ref{prove proposition1}) is the cosine similarity of the far-field channels, which can be verified by substituting $\mathbf{\Xi}_{m,n}=g_{m,n}\mathbf{I}_K$ and  $\mathbf{\Xi}_{m,n'}=g_{m,n'}\mathbf{I}_K$ into the left-hand side of (\ref{prove proposition1}).

\bibliographystyle{IEEEtran}
\bibliography{references}
\end{document}